\documentclass[sigconf,nonacm]{apr-arxiv}
\setcopyright{none}

\newcommand{\APRArxivVersion}{}

\AtBeginDocument{%
  }

\usepackage{booktabs}
\usepackage{multirow}
\usepackage{microtype}
\usepackage{listings}
\usepackage{etoolbox}

\newif\ifshownew
\shownewfalse

\usepackage[normalem]{ulem}
\definecolor{cutred}{rgb}{0.80,0.05,0.05}
\ifshownew
  \newcommand{\cut}[1]{{\color{cutred}\sout{#1}}}
\else
  \newcommand{\cut}[1]{}
\fi

\definecolor{newpurple}{rgb}{0.42,0.10,0.62}

\ifshownew
  \newcommand{\new}[1]{{\color{newpurple}#1}}
  \newenvironment{addedblock}[1]
    {\par\medskip\noindent{\color{newpurple}\rule{\linewidth}{1.2pt}}%
     \par\nobreak\smallskip\noindent
     {\color{newpurple}\footnotesize\textbf{[NEW: #1]}}%
     \par\nobreak\smallskip\begingroup\color{newpurple}}
    {\par\endgroup\smallskip\noindent
     {\color{newpurple}\rule{\linewidth}{1.2pt}}\par\medskip}
\else
  \newcommand{\new}[1]{#1}
  
\fi

\lstdefinestyle{appendixjson}{
  basicstyle=\ttfamily\footnotesize,
  breaklines=true,
  columns=fullflexible,
  keepspaces=true,
  showstringspaces=false,
  frame=single,
  framerule=0.3pt,
  xleftmargin=1em,
  xrightmargin=0.5em,
  aboveskip=0.75em,
  belowskip=0.75em,
  captionpos=b
}

\makeatletter
\AtEndPreamble{%
  \theoremstyle{acmdefinition}%
  \@ifundefined{remark}{}{}%
  \theoremstyle{acmplain}%
}
\makeatother

\begin{document}
\raggedbottom

\title[AutoPrunedRetriever]{AutoPrunedRetriever: Persistent Reasoning Graphs
  for Efficient and Extensible Retrieval-Augmented Generation}

\author{Ning Wang}
\authornote{Equal contribution.}
\email{nw366@cornell.edu}
\affiliation{%
  \institution{Cornell University}
  \city{Ithaca}
  \state{New York}
  \country{USA}}

\author{Kuanyan Zhu}
\authornotemark[1]
\email{kz345@cam.ac.uk}
\affiliation{%
  \institution{University of Cambridge}
  \city{Cambridge}
  \country{UK}}

\author{Daniel Yuehwoon Yee}
\authornotemark[1]
\email{u3636035@connect.hku.hk}
\affiliation{%
  \institution{The University of Hong Kong}
  \city{Hong Kong}
  \country{Hong Kong SAR}}

\author{Yitang Gao}
\affiliation{%
  \institution{The Hong Kong University of Science and Technology}
  \city{Hong Kong}
  \country{Hong Kong SAR}}

\author{Shiying Huang}
\affiliation{%
  \institution{Cornell University}
  \city{Ithaca}
  \state{New York}
  \country{USA}}

\author{Zirun Xu}
\affiliation{%
  \institution{University of British Columbia}
  \city{Vancouver}
  \state{British Columbia}
  \country{Canada}}

\author{Sainyam Galhotra}
\correspondingauthor
\email{sg@cs.cornell.edu}
\affiliation{%
  \institution{Cornell University}
  \city{Ithaca}
  \state{New York}
  \country{USA}}

\renewcommand{\shortauthors}{Wang et al.}

\begin{abstract}

Retrieval-augmented generation (RAG) has become a standard approach for
knowledge-intensive LLM tasks. Complex reasoning remains challenging because queries often require composing dispersed evidence through multiple dependent reasoning steps. Yet existing RAG systems process queries independently, repeatedly retrieving overlapping evidence and reconstructing established reasoning. Preserving such reasoning makes previously retrieved evidence and intermediate results reusable for subsequent queries. However, inference-time reuse becomes increasingly costly as the history grows, motivating a compact representation that preserves useful information without substantial loss in quality. To address these challenges, we present \textbf{AutoPrunedRetriever (APR)}, a graph-based RAG system that encodes corpus facts and query-derived reasoning in a shared symbolic codebook. The codebook serves both as the retrieval index and as persistent memory, representing facts, questions, and answers as compact edge sequences, so that the previously established structure can be directly retrieved and reused. At inference time, APR retrieves query-specific reasoning paths and constructs compact prompts, while incremental consolidation keeps the growing memory compact. For deployments where the retriever cannot be replaced, we introduce \textbf{AutoPruned Layer (APL)}, which leaves host retrieval unchanged but replaces the retrieved context with the corresponding entries in the shared symbolic codebook, writing each question and answer back as a compact structure. Across five complex-reasoning benchmarks, APR improves accuracy over the strongest baseline by 1.8--10.8 points while using up to 96\% fewer input tokens. APL reduces generator input by 60--96\% across host systems and preserves answer quality on document-style contexts.
\end{abstract}



\maketitle

\section{Introduction}
\label{sec:intro}

Retrieval-augmented generation (RAG) grounds LLMs in external knowledge, enabling knowledge-intensive tasks without requiring all information to be encoded in model parameters during training. Recent advances in retrieval and generation have shown strong performance on open-domain question answering and tool-augmented assistants~\citep{lewis2020rag,izacard2022atlas,ram2023in}. Yet complex reasoning remains challenging because answering a query often requires retrieving dispersed evidence and composing it through multiple dependent reasoning steps. Graph-based RAG methods improve this process by organizing evidence around entities and relations, making compositional structure easier to recover than from flat passages alone~\citep{edge2024local}.

The challenge becomes more pronounced when queries arrive as a stream. Existing RAG systems typically process each query independently, so related queries repeatedly retrieve overlapping evidence and reconstruct reasoning that earlier queries have already established. This repeated work increases retrieval costs, latency, and token usage and discards intermediate results that could directly support subsequent queries. In multi-turn settings, prior work reports that 80.34\% of retrieved content can consist of byte-exact duplicates~\citep{schelpe2026dedup}, illustrating how much information may be repeatedly recovered rather than reused. Preserving prior reasoning therefore provides a natural way to carry useful evidence and intermediate results across related queries.
However, retaining this history verbatim means every later reuse must process an ever-longer textual record, raising inference cost and making relevant information harder to recover~\citep{jiang2024longllmlinguaacceleratingenhancingllms}.
This motivates a representation in which the history grows with accumulated
knowledge rather than accumulated words, remaining compact enough to store and process efficiently.

Consider five documents describing the following facts about a company and its subsidiaries:
(i) Northstar acquired DeltaPay in March 2022;
(ii) DeltaPay is subject to the EU Digital Services Act;
(iii) DeltaPay recorded 12 GDPR incidents before the acquisition and 5 afterward;
(iv) CloudCo signed a 2024 vendor contract with DeltaPay; and
(v) DeltaPay became a Northstar subsidiary after the acquisition.
Suppose three related questions arrive in sequence.

\textbf{Q1:} Which subsidiaries that were acquired after Jan.\ 1, 2021 are subject to the EU Digital Services Act?

\textbf{Answer:} DeltaPay.
Answering Q1 establishes the acquisition, subsidiary, and regulatory relations needed to identify DeltaPay.

\textbf{Q2:} For those subsidiaries, did GDPR incident rates decrease post-acquisition?

\textbf{Answer:} Yes.
Q2 reuses the subsidiary and acquisition information established by Q1 and only needs to retrieve the additional incident evidence.

\textbf{Q3:} Which 2024 vendor contracts involve those same subsidiaries?

\textbf{Answer:} CloudCo's contract with DeltaPay.
Q3 again reuses Q1's identification of DeltaPay and only adds the vendor-contract relation.

Together, the three queries illustrate how related questions repeatedly reuse a shared core of established reasoning while adding only limited new information.

These observations motivate \textbf{AutoPrunedRetriever (APR)}, a structure-first RAG system built around a shared symbolic codebook of entities, relations, and edges. \textbf{For complex reasoning}, APR represents corpus facts as compact edge sequences and retrieves short reasoning paths directly in symbol space, so Q1 can compose the acquisition, subsidiary, and regulatory evidence without passing broad surrounding context to the generator. \textbf{For reuse across related queries}, questions and answers are written back into the same codebook, so the DeltaPay and acquisition information established by Q1 can be directly reused by Q2 and Q3 rather than reconstructed from scratch. \textbf{For compact history preservation}, the structure shared by Q1--Q3 is represented once through common entity, relation, and edge identifiers, while Q2 and Q3 add only their newly established incident and vendor-contract relations. Incremental entity consolidation further merges likely aliases and removes redundant structure as the memory grows. For deployments that already rely on a specialized retriever, we introduce the \textbf{AutoPruned Layer (APL)}. APL uses the same symbolic codebook and persistent memory as APR, but replaces APR's retrieval component with the host retriever. The host therefore continues to determine which chunks are retrieved, while APL replaces their textual payload with the corresponding structured codebook entries for generation and subsequent storage. Questions and answers are written back into the same memory, so later queries operate over a compact structured history rather than accumulated textual trajectories.

These design choices yield strong empirical results. Across five complex-reasoning benchmarks, APR improves accuracy over the strongest baseline on every dataset by 1.8--10.8 points while using up to 96\% fewer input tokens. On Medical~\citep{xiang2025when-graphs-in-rag}, APR improves over PathRAG~\citep{chen2026pathrag} from 0.678 to 0.707 while reducing generator input from 28k to 1.2k tokens. Its persistent state grows linearly with the query stream at 8.1--18.6\,KB per question over streams exceeding two thousand queries. APL reduces generator input by 60--96\% across host systems while preserving answer quality. When the host returns document-style evidence, replacing the textual payload with structured codebook entries can further improve answer quality. For G-Reasoner~\citep{luo2025g}, APL improves Medical complex-reasoning accuracy by over 13 points while reducing context length by about 67\%.

\section{Related Work}
\label{sec:related}

\subsection{Retrieval-Augmented Generation} 
Retrieval-augmented generation (RAG) grounds LLMs in external knowledge through retriever--reader pipelines~\citep{Chen2017DrQA,lewis2020rag}. REALM integrates retrieval into pre-training~\citep{Guu2020REALM}, FiD fuses passage-wise encodings for open-domain QA~\citep{Izacard2021FiD}, and Atlas demonstrates strong few-shot performance with updatable indices~\citep{izacard2022atlas}. Despite these advances, such text-first systems remain passage-centric and typically treat queries independently, which limits reuse and structural composition in multi-hop reasoning. Dense retrieval supported by ANN indexes~\citep{johnson2017billion,Malkov2016HNSW}, cross-encoder rerankers~\citep{Nogueira2019BERT}, and late-interaction models~\citep{Khattab2020ColBERT,ColBERTv2} improves ranking quality but does not address reasoning structure across documents.

\subsection{Graph-Based Retrieval}
Graph-based RAG methods instead retrieve over entities and relations: MS-GraphRAG builds corpus-level graphs and retrieves summarized subgraphs~\citep{edge2024local}; Lazy-GraphRAG defers graph construction to query time~\citep{Edge2024LazyGraphRAG}; HippoRAG and HippoRAG2 introduce neuro-inspired memory for reasoning paths~\citep{HippoRAG,HippoRAG2}; LightRAG and Fast-GraphRAG simplify graphs into key/value-style retrieval structures~\citep{guo2024lightrag,CircleMindAI2024FastGraphRAG}; RAPTOR uses recursive summary trees~\citep{Sarthi2024RAPTOR}; G-Reasoner learns subgraph relevance with a graph foundation model~\citep{luo2025g}; and PathRAG prunes retrieval to relational paths~\citep{chen2026pathrag}. These methods move beyond flat passages, but they still pay substantial graph serialization and token costs and generally do not persist reasoning structure across queries for reuse. APR and APL also differ by keeping retrieved evidence in the same compact entity-relation-entity representation used for persistent question-answer memory, rather than using graph structure only to select evidence for textual generation.

\subsection{Persistent Retrieval Memory and Context Compression}

Memory-augmented architectures maintain external stores alongside model parameters~\citep{wu2022memorizing,dao2023hungry}, building on earlier differentiable memory networks~\citep{weston2014memory} and neural Turing machines~\citep{graves2014neural}; efficient deployment techniques such as quantized adaptation~\citep{dettmers2023qlora}, mixed-precision training~\citep{micikevicius2017mixed}, and adaptive computation~\citep{schuster2022confident,li2023efficient} are complementary efforts to reduce latency and cost. APR combines compact representation and persistence in one shared codebook, while APL keeps the host retriever but replaces its retrieved textual payload with the same structured representation and stores subsequent question-answer history in that form.
\section{Method}
\label{sec:method}

\begin{figure*}[t]
  \centering
  \includegraphics[width=0.8\textwidth]{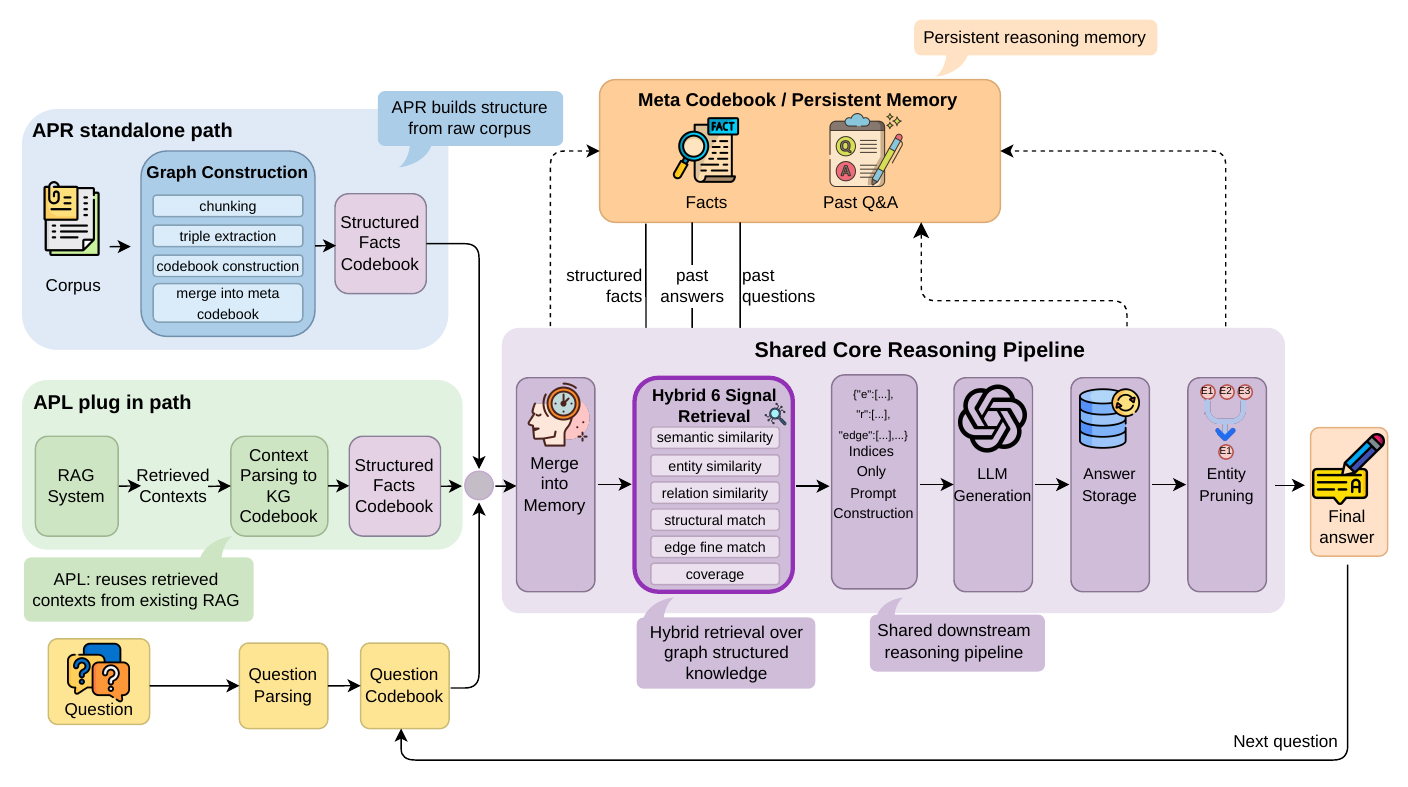}
  \caption{\textbf{AutoPrunedRetriever} pipeline}
  \Description{Pipeline diagram of AutoPrunedRetriever: text is parsed into triples and indexed into a symbolic meta-codebook, grouped into chunked small graphs, retrieved by a two-layer path retriever, filtered by a per-channel selector into a compact prompt, and periodically compacted by entity consolidation.}
  \label{fig:workflow}
\end{figure*}

\subsection{Problem Formulation}
\label{sec:problem}

We consider a corpus $\mathcal{D}$ and a sequence of queries
$q_1,\ldots,q_T$. Queries may share relevant evidence, entities, relations,
or intermediate reasoning. Before answering $q_t$, the system has a persistent
state $\mathcal{S}_{t-1}$ containing information retained from earlier queries.
It constructs a generator payload $\Pi_t$ from the current query, the corpus,
and the available state, produces an answer $a_t$, and may update the state to
$\mathcal{S}_t$. Independent-query RAG is the special case in which no
query-derived information is retained across queries.

Let $y_t$ denote the reference answer and let $|\Pi_t|$ denote the token length
of the generator payload. With the generator held fixed, we evaluate a system by its average answer
quality and average generator-input length:
\[
A_T =
\frac{1}{T}\sum_{t=1}^{T}\mathrm{Acc}(a_t,y_t),
\qquad
C_T =
\frac{1}{T}\sum_{t=1}^{T}|\Pi_t|.
\]
For long-running query streams, we additionally measure the growth of persistent
state as
\[
M_T =
\frac{|\mathcal{S}_T|-|\mathcal{S}_0|}{T},
\]
where $|\mathcal{S}_t|$ is its serialized size. The problem is to preserve high
answer quality with compact generator inputs while keeping the accumulated state
manageable as the query stream grows. APR instantiates $\mathcal{S}_t$ as
the symbolic codebook introduced below, and APL (Section~\ref{sec:apl}) shares
the same instantiation.

\subsection{Overview}
\label{sec:overview}

Guided by the design principles in Section~\ref{sec:intro}, APR operates in
five stages. First, it encodes questions, answers, and facts into a shared
symbolic codebook of entities and relations (Section~\ref{sec:encoding}).
Second, it groups these symbols into local reasoning graphs that serve as the
basic retrieval units (Section~\ref{sec:chunks}). Third, it retrieves
path-structured evidence in symbol space using a two-layer scoring procedure,
rather than expanding broad graph neighborhoods (Section~\ref{sec:retrieve}).
Fourth, it selects non-redundant structure and packs it into compact prompts (Section~\ref{sec:selection-prompt}).
Finally, it incrementally consolidates aliases and prunes low-value structure
to keep the persistent memory compact over time (Section~\ref{sec:consolidation}).
APL reuses the same symbolic memory, retrieval, and compression mechanisms as a
plug-in layer for existing RAG systems, as described in
Section~\ref{sec:apl}.

\subsection{Step 1: Symbolic Encoding}
\label{sec:encoding}

We first convert free-form text into a symbolic representation that makes
shared structure explicit and reuse exact. Instead of matching raw strings
across questions and facts, APR maps recurring entities, relations, and edges
to integer IDs, so later retrieval and reuse operate over compact symbolic
sequences. Formally, for a text span $y$ (a question, an answer, or a fact), an LLM-based parser
extracts triples
$\tau(y)\subseteq \mathcal{U}_E \times \mathcal{U}_R \times \mathcal{U}_E$.
These triples are stored in a persistent meta-codebook
\[
  \mathcal{C} = (E,\, R,\, \mathsf{M},\, \mathsf{Q},\, \mathsf{A},\, \mathsf{F},\,
  E_{\mathrm{emb}},\, R_{\mathrm{emb}},\, M_{\mathrm{emb}},\,
  Q_{\mathrm{emb}},\, A_{\mathrm{emb}},\, F_{\mathrm{emb}}),
\]
where $E$ and $R$ are entity and relation dictionaries,
$\mathsf{M}\subseteq E\times R\times E$ is the sparse set of unique edges, and
$\mathsf{Q}$, $\mathsf{A}$, and $\mathsf{F}$ store edge-ID sequences for
questions, answers, and facts. An \textsc{Indexify} map
$\mathbf{y} = \textsc{Indexify}(\tau(y); E, R, \mathsf{M}) \in
\mathsf{M}^{\star}$ converts each parsed text span into a sequence of edge
indices, extending $(E, R, \mathsf{M})$ only when new symbols or edges appear. To support retrieval at multiple granularities, we cache embeddings alongside
this symbolic structure. $E_{\mathrm{emb}}$ and $R_{\mathrm{emb}}$ represent
individual entities and relations, $M_{\mathrm{emb}}$ represents unique edges,
and $Q_{\mathrm{emb}}$, $A_{\mathrm{emb}}$, and $F_{\mathrm{emb}}$ represent
the full linearized sequences for questions, answers, and facts. This symbolic
encoding is both compact and retrieval-friendly: because natural language is
heavy-tailed, much of its informational mass can be captured by a relatively
small reusable vocabulary, yielding strong compression gains and reducing the token-dilution effect that
arises in long raw-text sequence embeddings.

\subsection{Step 2: Chunked Small Graphs}
\label{sec:chunks}

Once text has been encoded as edge sequences, APR organizes these edges into
coherent retrieval units. Rather than inserting every triple into a single
global graph, it builds small local graphs in a streaming pass, so retrieval
operates over compact reasoning segments rather than over a large shared
structure. For each modality $y\in\{\mathsf{Q}, \mathsf{A}, \mathsf{F}\}$, we maintain a
repository of runs $\mathsf{Runs}(y)\in(\mathsf{M}^{\star})^{\star}$. Each run
is constructed by maintaining a current small graph $G_k=(V_k,E_k)$ together
with an embedding centroid. For an incoming triple, we compute a fit score that
combines semantic cohesion, measured by cosine similarity between the triple
embedding and the centroid of $G_k$, with structural continuity, measured as a
bonus when the triple reuses nodes in $V_k$ or extends an existing path. If
this bonus-adjusted score exceeds a threshold $\tau_c$, the triple is appended to
$G_k$; otherwise, $G_k$ is closed, its edges are linearized into
$\mathbf{g}_k\in\mathsf{M}^{\star}$ and stored in $\mathsf{Runs}(y)$, and a new
graph $G_{k+1}$ is started. This produces maximal locally coherent runs. Because a single streaming pass can introduce spurious cuts near boundaries, we apply a local boundary correction step. For each proposed boundary, we re-embed
the nearby region, re-run the segmenter locally, and merge the adjacent runs if
the new segmentation no longer supports a cut at that position. The surviving
boundaries are therefore locally self-consistent fixed points. The final result is a sequence of compact,
coherent runs that serve as the basic retrieval atoms.

\subsection{Step 3: Two-Layer Path Retrieval}
\label{sec:retrieve}

Given a repository of coherent runs, APR needs to identify the few runs that support a query without exhaustively scoring the full repository or biasing retrieval toward long, noisy chunks. To do so, APR uses a two-layer retrieval scheme. Layer~1 builds a high-recall candidate pool using two complementary retrieval tracks over all runs, while Layer~2 reranks only that pool using fine-grained edge-level signals. The result is a coarse-to-fine retrieval procedure that combines efficiency with high recall. In Layer~1, Track~A captures global semantic similarity by combining three embedding-based signals with Reciprocal Rank Fusion ($k$=60): query--run sequence similarity, maximum pairwise entity similarity, and maximum pairwise relation similarity; Track~B captures structural overlap using an inverted index from edge IDs to runs. For each query edge, APR retrieves the top-$K$ most similar vocabulary edges ($K=\min(500,\lfloor\sqrt{|\mathsf{M}|}\rfloor)$), propagates their similarity scores to runs containing them, and normalizes each run's accumulated score by the square root of its edge count so that longer runs do not dominate. The candidate pool is the union of the top runs from both tracks, ensuring that neither global similarity nor structural match can veto the other. Layer~2 then computes two edge-level scores within this pool from the cosine matrix $\mathbf{S}$ between query and candidate edge embeddings. The first is an edge-fine score, which averages the top-$\sqrt{n_q \cdot n_c}$ entries of $\mathbf{S}$, where $n_q$ and $n_c$ are the numbers of query and candidate edges. The second is a coverage score, which averages, over query edges, the best cosine match in the candidate. A final 4-way RRF combines the Layer~1 global ranking, the Layer~1 structural ranking, and the two Layer~2 rankings to produce the final ordering. The top runs are then passed to the prompt-construction stage.


\subsection{Step 4: Selection and Compact Prompt Construction}
\label{sec:selection-prompt}

Retrieval may return overlapping runs that repeat the same evidence across prior
answers, corpus facts, and related questions. APR therefore selects evidence
within each channel before constructing the generator context, preventing
redundant structure from occupying the prompt and allowing complementary
evidence from different sources to be retained. Let
$\mathbf{q}$, $\mathbf{a}$, and $\mathbf{f}$ denote the selected edge sequences
from related questions, prior answers, and corpus facts. Their union
\(
\mathsf{U}=\mathbf{q}\cup\mathbf{a}\cup\mathbf{f}
\)
determines the entities $E'$ and relations $R'$ required by the current query.
Only this query-specific subset of the codebook is serialized for generation. APR supports two equivalent structured encodings of $\mathsf{U}$. The
\emph{edge-reference} form stores each unique edge once and lets questions,
answers, and facts refer to its edge ID. The \emph{inline-triple} form instead
represents each selected item directly as indexed entity--relation--entity
triples. APR constructs both forms and uses the one with fewer tokens for the
current query. The resulting payload
\(
\Pi=(E',R',\mathbf{q},\mathbf{a},\mathbf{f},\texttt{rules})
\)
contains only the vocabulary and evidence selected for that query, avoiding the
repeated natural-language realization of shared entities, relations, and edges.

\ifdefined\APRArxivVersion
\noindent\begin{minipage}[t]{\linewidth}
\else
\begin{minipage}[t]{0.9\linewidth}
\fi
\small
\textbf{(a) Edge-reference JSON schema}\\[-0.25em]
\begin{verbatim}
---Knowledge Base---
[JSON format]
- e: entity vocabulary, e[i] = entity string
- r: relation vocabulary, r[j] = relation string
- edges: [[head_e_idx, r_idx, tail_e_idx]]
- questions: edge IDs supporting related questions
- given knowledge: edge IDs supporting prior answers
- facts: edge IDs supporting corpus facts
\end{verbatim}
\end{minipage}

\vspace{0.5em}

\ifdefined\APRArxivVersion
\noindent\begin{minipage}[t]{\linewidth}
\else
\begin{minipage}[t]{0.9\linewidth}
\fi
\small
\textbf{(b) Inline-triple JSON schema}\\[-0.25em]
\begin{verbatim}
---Knowledge Base---
[JSON format]
- e: entity vocabulary, e[i] = entity string
- r: relation vocabulary, r[j] = relation string
- questions: [[head_e_idx, r_idx, tail_e_idx], ...]
- given knowledge: [[head_e_idx, r_idx, tail_e_idx], ...]
- facts: [[head_e_idx, r_idx, tail_e_idx], ...]
\end{verbatim}
\end{minipage}

\subsection{Step 5: Entity Consolidation}
\label{sec:consolidation}

As the persistent codebook grows, it may accumulate aliases and near-duplicate
entities, which increase retrieval redundancy and memory size. APR therefore
performs incremental consolidation rather than rebuilding the codebook from
scratch. It maintains a dirty entity set containing symbols introduced since
the previous consolidation round and only examines this subset when updating the
memory. For dirty entities, APR constructs an approximate nearest-neighbor graph over
entity embeddings and merges entities whose cosine similarity exceeds a global threshold $\tau_E$ into alias groups using
union-find. Each group is represented by a medoid entity, after which stored
edges are remapped as
\(
(u,r,v)\rightarrow(m_E(u),r,m_E(v)).
\)
Duplicate edges created by this remapping are removed, while cached embeddings
are updated only for affected edges and sequences. Because questions, answers,
and facts are represented as symbolic edge-ID sequences, consolidation updates
the shared memory through identifier remapping rather than rebuilding stored
trajectories. This quotient operation reduces redundant symbolic structure while preserving
the relations represented in the memory (Lemma~\ref{lem:quotient}). Moreover,
incremental maintenance avoids revisiting unaffected portions of the codebook,
with update cost determined by the changed subset rather than the full memory
(Lemma~\ref{lem:vec-cost}).

\begin{lemma}[Symbolic consolidation is monotonic]
\label{lem:quotient}
Let $\pi:E\rightarrow E/{\sim}$ map each entity to its consolidation
representative, let
$\phi(u,r,v)=(\pi(u),r,\pi(v))$, and write
$\mathsf{M}'=\mathrm{uniq}(\phi(\mathsf{M}))$ for the remapped, deduplicated
edge set.
After remapping and deduplication,
\(
|\mathsf M'|\leq |\mathsf M|,
\)
and for every stored sequence
$\sigma\in\mathsf Q\cup\mathsf A\cup\mathsf F$,
\(
|\mathrm{uniq}(\phi(\sigma))|\leq |\sigma|.
\)
Thus consolidation can remove redundant structure but cannot increase the
symbolic memory size.
\end{lemma}

\begin{proof}
The remapping function $\phi$ maps each original edge to a possibly shared
representative edge. Deduplication removes repeated edges produced by this
mapping, so the number of unique edges cannot exceed the original number of
edges. The same argument applies to stored edge sequences: remapping can merge
identical edges but cannot create additional unique elements.
\end{proof}

\begin{lemma}[Incremental consolidation cost]
\label{lem:vec-cost}
Let $D_e$ denote the set of edges whose symbolic representation changes after entity
remapping, and let $D_c$ denote the set of stored sequences containing at least one edge in
$D_e$. With cached embeddings, consolidation requires updating only
$|D_e|$ edge embeddings and $|D_c|$ sequence embeddings, rather than all stored
edges and sequences.
\end{lemma}

\begin{proof}
Entity remapping changes only edges containing consolidated entities.
Therefore, edge embeddings remain valid for all edges outside $D_e$. A stored
question, answer, or fact sequence changes only when it contains an affected
edge, so only sequences in $D_c$ require recomputation. All other cached
representations can be reused.
\end{proof}

\subsection{AutoPruned Layer (APL)}
\label{sec:apl}

APR provides a complete retrieval and memory system, but many existing RAG
deployments already rely on specialized retrievers that they do not wish to
replace. APL extends these systems by separating evidence selection from
evidence representation. It keeps the host retriever unchanged and replaces the
textual interface between retrieval and generation with the same symbolic
representation used by APR. APL first constructs an APR codebook over the host corpus using the same chunking
and symbolic encoding process as APR. Since retrieval operates over the same
chunks, each retrieved chunk directly corresponds to entries in the codebook.
At query time, the host retriever selects relevant chunks, and APL maps these
selections to their entity--relation--entity representations rather than
passing the original text to the generator. The selected evidence is combined
with relevant question--answer structures retrieved from persistent memory and
serialized into a compact structured prompt. After generation, the current question and answer are encoded into the same
symbolic memory and become available for future reuse. Entity consolidation
maintains the compactness of this growing state. Consequently, APL is able to augment existing RAG pipelines by converting
retrieved text into reusable symbolic structure.
\section{Experiments}
\label{sec:exp}

\subsection{Experimental Setup}
\label{sec:exp-setup}

\textbf{Datasets}.
We evaluate on five complex-reasoning benchmarks. \textbf{Medical} and
\textbf{Novel} come from GraphRAG-Benchmark~\citep{xiang2025when-graphs-in-rag},
\textbf{STEM} and \textbf{TV} are constructed from HotpotQA Wikipedia
micro-corpora~\citep{hotpotqa-wiki}, and \textbf{Quality} comes from
QuALITY~\citep{pang2022quality}. STEM and TV stress non-local composition by
interleaving distractor text between the supporting spans of multi-hop
questions, so the required evidence is never adjacent. TV contains 102
questions over a 1.04M-word corpus, while STEM contains 118 questions over a
0.48M-word corpus with exactly two supporting spans per question. The two are complementary: TV is a heterogeneous, variable-span setting with 1--5 supporting spans per question, whereas STEM is a controlled, domain-coherent setting with exactly two spans per question. We abbreviate the complex-reasoning and fact-retrieval tasks as CR and FR.

\textbf{Baselines}.
We compare against HippoRAG2~\citep{HippoRAG2},
LightRAG~\citep{guo2024lightrag}, G-Reasoner~\citep{luo2025g}, and
PathRAG~\citep{chen2026pathrag}. All systems use GPT-4o-mini as the generator
and BAAI/bge-large-en-v1.5 for embeddings where applicable. We follow the
GraphRAG-Benchmark configurations for G-Reasoner and PathRAG and use the
released G-Reasoner retriever. On Medical and Novel, we use the published
GraphRAG-Benchmark results for HippoRAG2 and LightRAG. All remaining baseline
results are produced under the same evaluation protocol. Experiments run on a
single NVIDIA RTX 4090.

\textbf{Evaluation protocol}.
We measure answer quality with the GraphRAG-Benchmark LLM judge~\citep{xiang2025when-graphs-in-rag} (answer correctness, ACC, in $[0,1]$; GPT-4o-mini as judge)
 and report the actual query-time input tokens
passed to the generator. Memory accumulates within each benchmark in its native
question order, with alternative orderings evaluated separately for robustness.
Each APR and APL hyperparameter has a single global setting across datasets and
host systems unless explicitly varied in sensitivity analyses. For APL, each
host and its host+APL counterpart use identical retrieved chunk IDs. The host therefore
selects exactly the same evidence, while APL changes only its representation
for generation and the persistent history by mapping each selected chunk to
its precomputed entries in the shared APR codebook.

\subsection{RQ1: End-to-End Effectiveness}
\label{sec:rq1}

\paragraph{Complex-reasoning performance and token efficiency.}
Table~\ref{tab:cr_and_tokens} compares complex-reasoning accuracy and
query-time generator input across all five datasets. We first consider the
three published benchmarks. On \textbf{Novel}, APR reaches \textbf{0.642},
improving over the strongest baseline, HippoRAG2, by \textbf{10.8} points while
using 938 rather than 1.5K input tokens. On \textbf{Medical}, APR reaches
\textbf{0.707}, exceeding PathRAG's 0.678 while reducing generator input from
28K to 1.2K tokens. It also improves over HippoRAG2 on both accuracy and input
length, from 0.620 to 0.707 and from 1.5K to 1.2K tokens. On
\textbf{Quality}, APR achieves the highest accuracy at \textbf{0.323}, compared
with 0.304 for HippoRAG2, while using less than half as many input tokens.
Together, these results show that APR's gains extend across published
benchmarks with different corpora and reasoning settings rather than depending
on the two datasets constructed for this study. We next evaluate the controlled \textbf{STEM} and \textbf{TV} benchmarks, where
supporting evidence is deliberately separated by distractor text. APR again
achieves the highest accuracy on both. On STEM, it reaches \textbf{0.776},
improving over HippoRAG2 by \textbf{7.7} points, while using a comparable
generator input of 2.1K versus 1.9K tokens. On TV, APR reaches \textbf{0.648},
ahead of G-Reasoner's 0.630, while reducing generator input from 5.1K to 687
tokens. Across all five benchmarks, APR therefore improves over the strongest
accuracy baseline on every dataset by 1.8--10.8 points while using up to
96\% fewer generator-input tokens.

\begin{table}[t]
\centering
\caption{Complex-reasoning accuracy and query-time generator input across five
datasets. Best values are in \textbf{bold}. Format is
\textit{accuracy / tokens}. Token counts include the complete prompt passed to
the generator, including structured evidence and prompt scaffolding.
\textsuperscript{*} denotes a retrieval failure followed by a question-only
prompt.}
\label{tab:cr_and_tokens}
\resizebox{\columnwidth}{!}{%
\begin{tabular}{lccccc}
\toprule
\textbf{System} & \textbf{Medical} & \textbf{Novel} & \textbf{STEM} &
\textbf{TV} & \textbf{Quality} \\
\midrule
\textbf{APR}    & \textbf{.707}\,/\,\textbf{1.2k}
                & \textbf{.642}\,/\,\textbf{938}
                & \textbf{.776}\,/\,2.1k
                & \textbf{.648}\,/\,\textbf{687}
                & \textbf{.323}\,/\,\textbf{727} \\
HippoRAG2       & .620\,/\,1.5k
                & .534\,/\,1.5k
                & .699\,/\,\textbf{1.9k}
                & .595\,/\,1.6k
                & .304\,/\,1.5k \\
PathRAG         & .678\,/\,28k
                & .502\,/\,27k
                & .612\,/\,24k
                & .496\,/\,30k
                & .303\,/\,27k \\
G-Reasoner      & .575\,/\,5.1k
                & .488\,/\,5.1k
                & .625\,/\,3.4k
                & .630\,/\,5.1k
                & .299\,/\,3.9k \\
LightRAG        & .613\,/\,100k
                & .491\,/\,101k
                & .171\,/\,8.8k
                & .465\,/\,3.0k
                & .196\,/\,96\textsuperscript{*} \\
\bottomrule
\end{tabular}%
}
\end{table}


\paragraph{Performance across the full published task suite.}
The results above focus on complex reasoning, where APR is designed to be most
useful. We next evaluate whether this advantage extends beyond that task.
Table~\ref{tab:rag_comparison} reports the complete GraphRAG-Benchmark suite on
the published \emph{Novel} and \emph{Medical} corpora, covering fact retrieval,
complex reasoning, context summarization, and creative generation. APR obtains the highest average accuracy on both corpora. On \emph{Novel}, it
reaches \textbf{65.0}, improving over the next-best system, HippoRAG2, by
\textbf{8.5} points. The gain is not limited to complex reasoning. APR achieves
the highest Complex Reasoning accuracy at \textbf{64.2}, Context Summarization
accuracy at \textbf{81.0}, and Creative Generation accuracy at \textbf{65.5}.
Its 938-token average generator input is also close to vanilla RAG and
substantially smaller than high-context graph systems such as MS-GraphRAG and
LightRAG. On \emph{Medical}, where several baselines are already strong, APR obtains the
highest average accuracy at \textbf{65.7}, narrowly exceeding HippoRAG2 at
64.9. The clearest advantage again appears on Complex Reasoning, where APR
reaches \textbf{70.7} compared with HippoRAG2's 62.0 and Fast-GraphRAG's 61.7.
For the remaining tasks, APR is competitive rather than uniformly dominant,
including second-best Context Summarization accuracy at 67.6. These results
show that APR's strongest gains occur on compositional reasoning while its
overall performance remains competitive across the broader task suite.

\begin{table*}[!t]
\centering
\caption{GraphRAG-Benchmark results on \emph{Novel} and \emph{Medical}.
Best values are in \textbf{bold} and second-best values are
\underline{underlined}. ACC: LLM-judged answer accuracy; R-L: ROUGE-L; Cov: evidence coverage; FS: faithfulness; Avg ACC: macro-average of the four task accuracies; Tok: average generator-input tokens.}
\label{tab:rag_comparison}
\renewcommand{\arraystretch}{1.22}
\setlength{\tabcolsep}{7.5pt}
\normalsize
\begin{tabular}{llccccccccccc}
\toprule
& \textbf{Model} &
\multicolumn{2}{c}{\textbf{Fact Ret.}} &
\multicolumn{2}{c}{\textbf{Complex Reas.}} &
\multicolumn{2}{c}{\textbf{Ctx.\ Summ.}} &
\multicolumn{3}{c}{\textbf{Creative Gen.}} &
\multicolumn{2}{c}{\textbf{Avg}} \\
\cmidrule(lr){3-4}\cmidrule(lr){5-6}\cmidrule(lr){7-8}
\cmidrule(lr){9-11}\cmidrule(lr){12-13}
& & ACC & R-L & ACC & R-L & ACC & Cov & ACC & Cov & FS &
\textbf{ACC} & \textbf{Tok} \\
\midrule
\multirow{10}{*}{\rotatebox{90}{\small\textbf{Novel}}}
& RAG (no rerank) & 58.8 & \textbf{37.4} & 41.4 & 15.1 & 50.1 & 82.5 & 41.5 & 47.5 & 37.8 & 48.0 & \textbf{879} \\
& RAG (rerank) & \textbf{60.9} & 36.1 & 42.9 & 15.4 & 51.3 & 83.6 & 38.3 & 49.2 & \textbf{40.0} & 48.3 & \textbf{879} \\
& MS-GraphRAG & 49.3 & 26.1 & 50.9 & 24.1 & \underline{64.4} & 75.6 & 39.1 & 55.4 & 35.7 & 50.9 & 38.7K \\
& HippoRAG & 52.9 & 26.7 & 38.5 & 11.2 & 48.7 & \textbf{85.6} & 38.9 & \textbf{71.5} & 39.0 & 44.8 & 7,208 \\
& HippoRAG2 & \underline{60.1} & 31.4 & \underline{53.4} & \underline{33.4} & 64.1 & 70.8 & \underline{48.3} & 49.8 & 31.0 & \underline{56.5} & 1,008 \\
& LightRAG & 58.6 & 35.7 & 49.1 & 24.2 & 48.9 & 63.1 & 23.8 & 57.3 & 25.0 & 45.1 & 101K \\
& Fast-GraphRAG & 57.0 & 35.9 & 48.6 & 21.1 & 56.4 & 80.8 & 46.2 & 57.2 & 37.0 & 52.1 & 4,204 \\
& RAPTOR & 49.3 & 23.7 & 38.6 & 11.7 & 47.1 & 82.3 & 38.0 & \underline{70.9} & 35.9 & 43.3 & 3,441 \\
& Lazy-GraphRAG & 51.7 & \underline{37.0} & 49.2 & 23.5 & 58.3 & 76.9 & 43.2 & 50.7 & \underline{39.7} & 50.6 & -- \\
& \textbf{APR} & 49.3 & 30.5 & \textbf{64.2} & \textbf{36.3} & \textbf{81.0} & \underline{83.7} & \textbf{65.5} & 34.2 & 32.1 & \textbf{65.0} & \underline{938} \\
\midrule
\multirow{10}{*}{\rotatebox{90}{\small\textbf{Medical}}}
& RAG (no rerank) & 63.7 & 29.2 & 57.6 & 14.0 & 63.7 & 77.3 & 58.9 & 35.9 & 57.9 & 61.0 & \textbf{954} \\
& RAG (rerank) & \underline{64.7} & 30.8 & 58.6 & 15.6 & 65.8 & \textbf{78.5} & 60.6 & 36.7 & \underline{58.7} & 62.4 & \textbf{954} \\
& MS-GraphRAG & 38.6 & 26.8 & 47.0 & 22.0 & 41.9 & 23.0 & 53.1 & 32.7 & 39.4 & 45.2 & 39.8K \\
& HippoRAG & 56.1 & 21.0 & 55.9 & 13.6 & 59.9 & 62.7 & 64.4 & \textbf{69.2} & \textbf{65.6} & 59.1 & 7,342 \\
& HippoRAG2 & \textbf{66.3} & 36.7 & \underline{62.0} & \textbf{37.0} & 63.1 & 46.1 & \textbf{68.1} & \underline{58.8} & 51.5 & \underline{64.9} & \underline{1,020} \\
& LightRAG & 63.3 & \textbf{37.2} & 61.3 & 25.0 & 63.1 & 51.2 & -- & -- & -- & -- & 100K \\
& Fast-GraphRAG & 60.9 & 31.0 & 61.7 & 21.4 & \textbf{67.9} & 52.1 & \underline{65.9} & 56.1 & 44.7 & 64.1 & 4,298 \\
& RAPTOR & 54.1 & 17.9 & 53.2 & 11.7 & 58.7 & \underline{78.3} & -- & -- & -- & -- & 3,510 \\
& Lazy-GraphRAG & 60.3 & 31.7 & 47.8 & 22.7 & 57.3 & 55.9 & 62.2 & 31.0 & 43.8 & 56.9 & -- \\
& \textbf{APR} & 61.6 & \underline{37.1} & \textbf{70.7} & \underline{31.8} & \underline{67.6} & 42.8 & 62.9 & 30.3 & 40.1 & \textbf{65.7} & 1,222 \\
\bottomrule
\end{tabular}
\end{table*}

\subsection{RQ2: What Drives the Gains?}
\label{sec:rq2}

\paragraph{Component ablation.}
Figure~\ref{fig:ablation} separates the effects of structured representation,
persistent memory, and entity consolidation across four datasets. Structured
representation has the clearest effect on Novel, Medical, and STEM. Replacing
the structured KG prompt with plain text reduces accuracy by \textbf{0.168} on
Novel, \textbf{0.040} on Medical, and \textbf{0.025} on STEM, while having
a negligible effect on TV (\textbf{+0.002}). Persistent memory provides a
complementary gain. Removing it lowers accuracy by \textbf{0.039} on Medical
and \textbf{0.031} on TV, with a smaller effect on Novel (\textbf{0.005}). STEM provides the opposite case. Its questions are constructed from
disjoint evidence spans with little cross-query reuse, and removing memory does
not reduce accuracy. The memory and consolidation variants on STEM also fall
within the 0.775--0.798 variation observed across query orderings. These results show that the benefit of persistent memory depends strongly on how much information earlier questions and answers share with subsequent queries, as well as on the quality of the retained information.  Entity consolidation has only a
modest effect on answer accuracy, with the largest drop of \textbf{0.015} on
Medical, consistent with its primary role in maintaining a compact persistent
state rather than directly improving reasoning. Overall, the ablation separates
two main sources of gain. Structured representation improves how retrieved
evidence is presented for reasoning, while persistent memory helps when later
queries can reuse information established earlier.

\begin{figure}[t]
  \centering
  \includegraphics[width=\columnwidth]{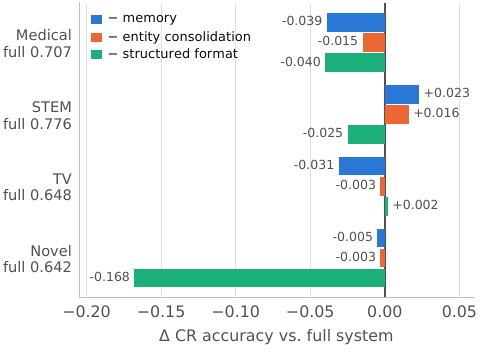}
  \caption{Cross-dataset ablation on complex reasoning: change in accuracy when each component is removed, with the full-system accuracy above each group.}
  \Description{Grouped bars of the change in complex-reasoning accuracy when memory, entity consolidation, or structured format is removed, for Medical, STEM, TV, and Novel. Removing structured format on Novel is the largest drop at minus 0.168; removing memory costs about 0.03 to 0.04 on Medical and TV; STEM deltas are small and slightly positive for memory and consolidation.}
  \label{fig:ablation}
\end{figure}

\paragraph{Representation transfer with host retrieval held fixed.}
APL tests whether the same symbolic interface remains useful when evidence is
selected by an external retriever rather than APR. For each host and host+APL
pair, we hold the retrieved chunk IDs fixed. APL maps those chunks to their
precomputed entries in the shared codebook and supplies the resulting structured
evidence to the generator, while subsequent questions and answers accumulate in
the same persistent memory. We run APL in two memory scopes: \emph{Full-mem} accumulates memory over the complete benchmark stream, including the fact-retrieval questions that precede the complex-reasoning set, whereas \emph{CR-only} accumulates memory over the complex-reasoning questions alone. Figure~\ref{fig:apl} shows substantial context reduction across both host
systems. APL reduces generator input by
\textbf{60--75\%} for G-Reasoner and by \textbf{89--96\%} for PathRAG.
CR-only is consistently more compact than Full-mem because its memory is smaller. The effect on answer quality depends on the representation produced by the host
retriever. G-Reasoner identifies relevant entities with a query-dependent GNN
but supplies document-level evidence to the generator. Replacing this textual
payload with its symbolic representation raises Medical accuracy from 0.576 to
\textbf{0.702} with Full-mem and \textbf{0.708} with CR-only, gains of
\textbf{12.6} and \textbf{13.2} points. On Novel, accuracy increases from
0.488 to 0.514 and 0.519. PathRAG already provides pruned relational paths, so
APL primarily reduces context size there, while still increasing Novel from
0.479 to \textbf{0.511} with Full-mem. The contrast suggests that APL provides
the largest accuracy gains when the host ultimately exposes document-style
context, while its reduction in generator input persists even when the host
already returns structured evidence.

\ifdefined\APRArxivVersion
\begin{figure*}[t]
\else
\begin{figure}[t]
\fi
  \centering
  \includegraphics[width=\linewidth]{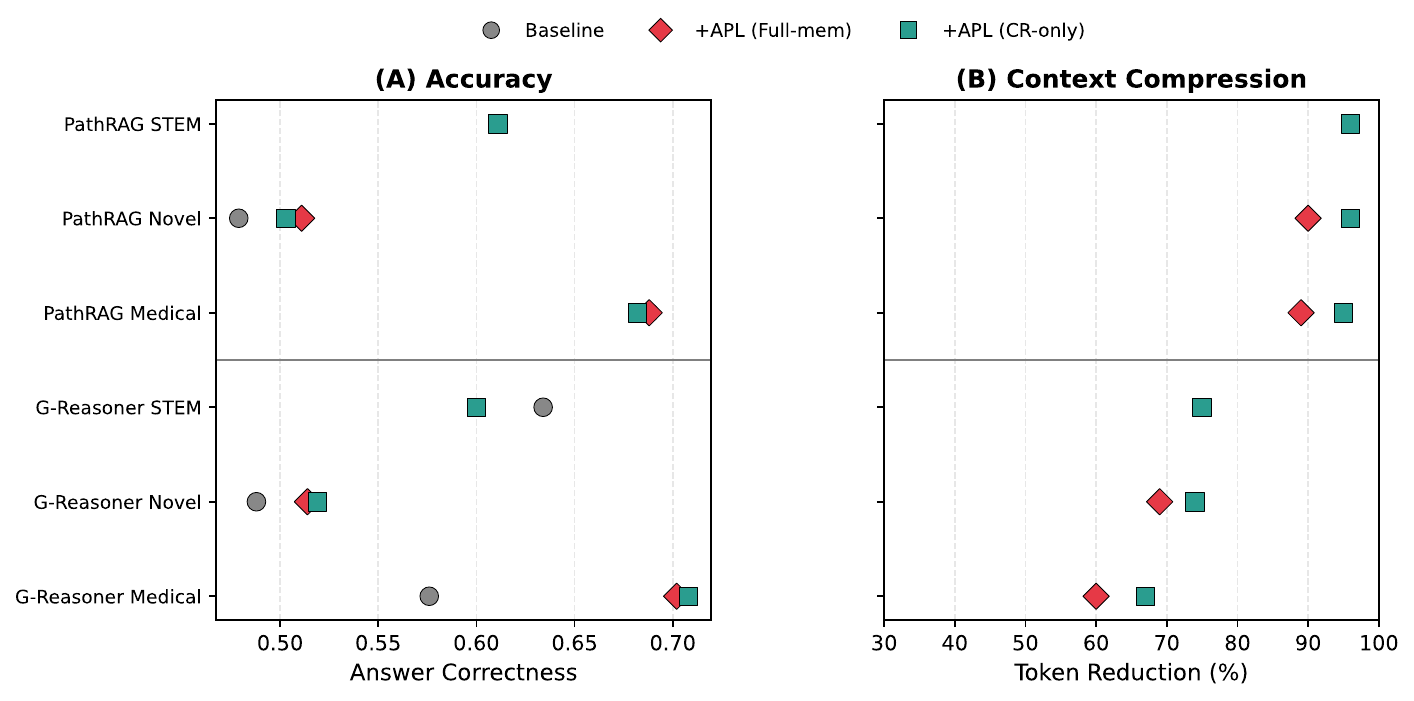}
  \caption{APL with host retrieval held fixed: compression (right) is uniform across hosts; accuracy (left) improves where the host's context is document-shaped.}
  \Description{Dot plot comparing AutoPruned Layer transfer accuracy and context compression when attached to different host RAG systems across datasets; compression is consistently large for every host and accuracy gains are largest for G-Reasoner.}
  \label{fig:apl}
\ifdefined\APRArxivVersion
\end{figure*}
\else
\end{figure}
\fi

\paragraph{Case Studies: Reasoning Format and Memory in Practice.} The preceding quantitative results give two complementary takeaways. The ablation identifies structured reasoning format and persistent memory as the components carrying APR's gains, while the APL results show that structured compression of retrieved evidence is what transfers when APR is layered on top of external RAG systems. To examine the \textbf{effect of reasoning format}, we compare G-Reasoner with G-Reasoner + APL on the Medical-CR question \emph{``How does the grade of follicular lymphoma influence its classification and treatment approach?''} Here both systems begin from the same retrieved evidence, but G-Reasoner consumes it as a long document-style passage that mixes the relevant grading facts with substantial background material on lymphocytes, B cells, and germinal centers. Its answer therefore remains factually rich but poorly organized, emphasizing descriptive biology rather than the clinically relevant mapping from grade to classification and treatment. APL serves the same context as explicit relational edges such as \texttt{Grade 3B} $\xrightarrow{\text{treated as}}$ \texttt{DLBCL} and \texttt{FL grade 3A} $\xrightarrow{\text{might be treated as}}$ \texttt{classic FL or DLBCL}, reducing the prompt from 5.1K to 1.4K tokens and improving the answer score from 0.502 to 0.965. This is a representation gain rather than a retrieval gain, matching the broader APL pattern that the method yields the largest gains where the host retriever returns document-style context. To examine the \textbf{effect of memory}, we compare APR with and without memory on the Medical-CR question \emph{``What is the sequence of genetic and chromosomal events leading to CML?''} The target answer requires a specific bridge---chromosome 9/22 translocation $\rightarrow$ Philadelphia chromosome $\rightarrow$ BCR::ABL1 $\rightarrow$ CML. In the no-memory setting, the retriever returns only generic edges about chromosomal translocations and rearrangements, so the model falls back to a superficially related example, \texttt{RUNX1::RUNX1T}, yielding 0.223. In the full-memory setting, earlier fact-retrieval questions have already stored the CML-specific bridge facts in the meta-codebook, and these are retrieved as explicit symbolic evidence when the later complex-reasoning question is asked; the resulting answer correctly names the Philadelphia chromosome, the chromosome 9/22 translocation, the BCR::ABL1 fusion gene, and the downstream disease mechanism, raising the score to 0.985. A second memory case isolates direct reuse of a previously established fact. For a bone-cancer question asking which performance-status scales guide treatment planning, the retrieved evidence mentions performance status but never names the scales, and APR without memory scores 0.204. Earlier esophageal- and mouth-cancer questions had stored \texttt{ECOG performance status} and \texttt{Karnofsky Performance Status} as edges; retrieving them raises the full-memory score to 0.980. Where the CML case shows memory supplying an intermediate reasoning chain, this case shows it recovering a precise detail established by earlier queries. Together, these cases connect back to the quantitative results: the APL example shows why structured reasoning format transfers, while the APR example shows how persistent memory carries forward bridge facts across a query stream. Further cases and the full GraphRAG-Benchmark task breakdown are included in the release.

\subsection{RQ3: Robustness, Scalability, and Cost}
\label{sec:rq3}

\begin{figure*}[t]
\centering
\includegraphics[width=\textwidth]{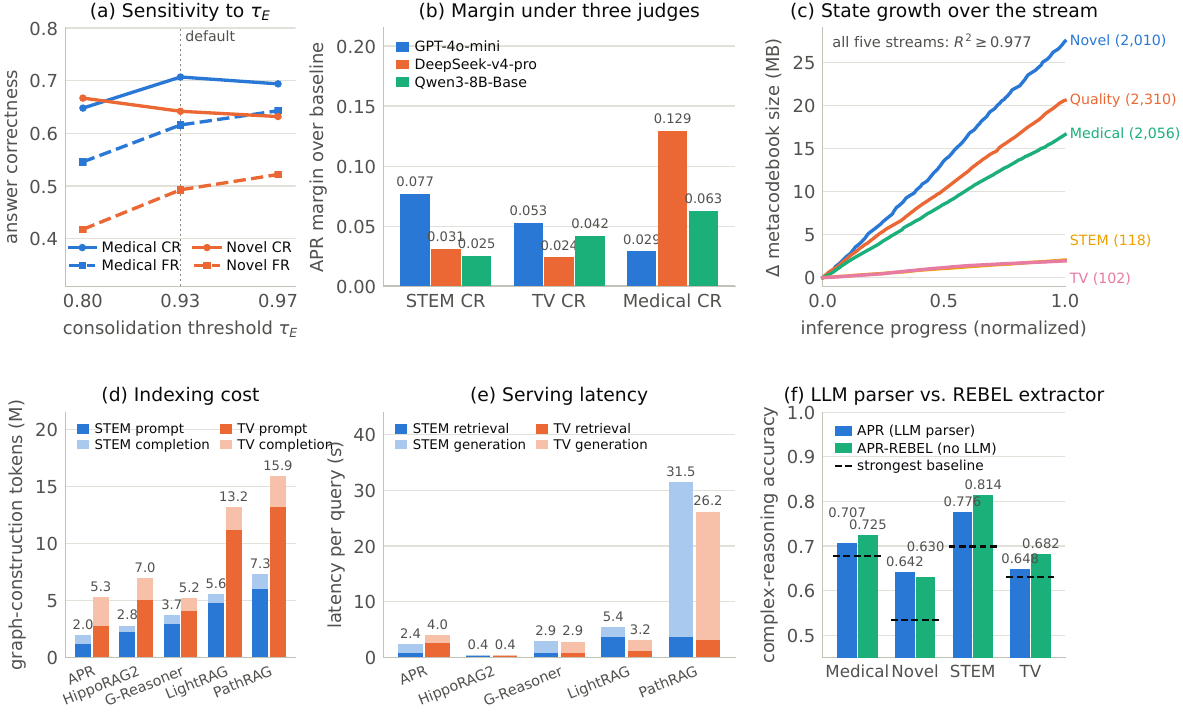}
\caption{Robustness, scalability, and cost. (a) Sensitivity to $\tau_E$. (b) Margin over the strongest baseline under three judges. (c) Persistent-state growth over full streams ($R^2\geq0.977$). (d) Graph-construction tokens and (e) per-query serving latency on STEM and TV. (f) LLM-parsed APR vs.\ APR-REBEL against the strongest baseline.}
\Description{Six panels. Top row: answer correctness against the consolidation threshold; grouped bars of APR's margin over the strongest baseline under three judge models, all above zero; metacodebook growth in megabytes over normalized inference progress for all five benchmarks, every curve a straight line. Bottom row: stacked graph-construction prompt and completion tokens for five systems on STEM and TV, APR lowest; stacked retrieval and generation latency per query for the same systems, PathRAG far above the rest; complex-reasoning accuracy of APR with the LLM parser versus the REBEL extractor on four datasets, both above the strongest baseline marker.}
\label{fig:regimes_robustness}
\end{figure*}

\paragraph{Robustness of the ranking.}
We test whether the end-to-end results depend on the evaluator, query order, or
triple extractor. Figure~\ref{fig:regimes_robustness}b re-scores the same
prediction files with GPT-4o-mini, DeepSeek-v4-pro, and Qwen3-8B-Base. APR
remains above the strongest non-APR baseline in every tested setting, with
margins of 0.025--0.129 under the alternative judges. Query order has similarly
limited impact. Reversing the stream and evaluating two random permutations
change accuracy by at most 3.0\% relative to its mean, with fact-retrieval
accuracy of $0.629\pm0.002$ on Medical and $0.511\pm0.007$ on Novel.
Finally, replacing the GPT-4o-mini triple parser with GPT-4.1-nano changes
end-to-end accuracy by only $-0.2\%$ to $+3.4\%$, while 99.59\% of emitted
lines on the full Medical corpus parse successfully into triples. The main
ranking is therefore stable across judge, query-order, and extraction-model
perturbations.

\paragraph{Behavior over long query streams.}
Persistent memory is useful only if its growth remains controlled without
degrading answer quality. Figure~\ref{fig:regimes_robustness}c measures the
persistent-state growth $M_T$ defined in Section~\ref{sec:problem}. Across all
five benchmarks, state size grows approximately linearly with the number of
processed queries, with $R^2\geq0.977$ and an incremental footprint of
8.1--18.6\,KB per question. We observe no late acceleration as the streams
grow. On the longest Medical stream of 2{,}056 questions, rolling
per-100-question accuracy is also stable. The fitted slopes are $-0.0002$ for
fact retrieval and $+0.0005$ for complex reasoning per 100 questions, with
both 95\% confidence intervals containing zero. Thus, increasing
history length produces predictable state growth without measurable degradation
in answer quality over the observed stream.

\paragraph{Parameter sensitivity.}
APR uses one consolidation threshold across all experiments rather than tuning
it per dataset. Figure~\ref{fig:regimes_robustness}a evaluates
$\tau_E\in\{0.80,0.93,0.97\}$. Complex-reasoning accuracy on STEM and TV varies
by at most 0.010 across the sweep, while the global $0.93$ setting gives the
highest complex-reasoning accuracy on Medical and TV. More conservative
consolidation improves fact retrieval, increasing accuracy on Medical from 0.616 to 0.643
and on Novel from 0.493 to 0.522 as $\tau_E$ rises to 0.97. The $0.93$ default
therefore provides a stable operating point across fact retrieval and complex
reasoning without dataset-specific tuning.

\paragraph{Indexing and serving cost.}
Query-time prompt size captures only part of the cost of graph-based RAG because
each system must first construct its retrieval index. Among the graph-based
systems for which we measure indexing cost, APR uses the fewest construction
prompt tokens on both STEM and TV (Figure~\ref{fig:regimes_robustness}d). It consumes 1.2M prompt and 0.8M completion tokens
on STEM and 2.8M prompt and 2.5M completion tokens on TV, compared with
6.0M/1.3M and 13.2M/2.7M for PathRAG and 4.8M/0.8M and 11.2M/2.0M for
LightRAG. APR therefore reduces construction prompt tokens relative to PathRAG
by approximately \textbf{80\%} on STEM and \textbf{79\%} on TV. The advantage
also appears at serving time. On STEM, APR requires 0.81\,s for retrieval and
1.62\,s for generation per query, compared with 31.5\,s end to end for
PathRAG (4.0\,s versus 26.2\,s on TV; Figure~\ref{fig:regimes_robustness}e), showing that the smaller representation translates into lower observed
serving latency rather than shifting cost elsewhere in the pipeline.

\paragraph{LLM-free symbolic encoding.}
APR does not require an LLM to construct its symbolic representation.
APR-REBEL replaces the LLM-based triple parser with the pretrained REBEL
extractor~\citep{huguet-cabot-navigli-2021-rebel-relation}, eliminating LLM
tokens during indexing (Figure~\ref{fig:regimes_robustness}f). It reaches \textbf{0.814} on STEM, 0.725 on Medical,
0.682 on TV, and 0.630 on Novel. Relative to the strongest non-APR \new{baseline} on
each corpus, these correspond to gains of \textbf{11.5}, \textbf{4.7},
\textbf{5.2}, and \textbf{9.6} points. Relative to the default LLM-parsed APR,
REBEL improves accuracy on STEM, Medical, and TV and remains within 0.012 of APR on Novel. It also
reduces query-time input from 2{,}144 to 1{,}090 tokens on STEM and from 687 to
523 on TV. This pattern suggests that a fixed local extractor can be effective
when its relation inventory matches the domain, while the LLM parser remains
useful for open or implicit relations. The result also separates APR's symbolic interface from its extractor: accuracy is preserved within 0.012 or improved when the LLM parser is replaced by a local relation extractor, so the gains do not depend on a particular LLM-based encoding procedure.

\section{Conclusion}

We presented \textsc{AutoPrunedRetriever} (APR), a RAG architecture that represents
corpus evidence and accumulated question--answer history in a shared symbolic
codebook. Rather than using structure only to retrieve evidence and then expanding
it back into text, APR retrieves compact reasoning paths, presents the selected
evidence to the generator in the same structured representation, and writes useful
reasoning from previous queries back into that representation for later reuse.
Across five complex-reasoning benchmarks, APR improves over the strongest baseline
on every dataset by 1.8--10.8 points while using up to 96\% fewer generator-input
tokens, and it obtains the highest average accuracy on both published
GraphRAG-Benchmark corpora.

Our analysis separates the two mechanisms behind these gains. Structured
representation makes retrieved relations easier to use for the current query,
while persistent memory helps when later queries can reuse information established
earlier. APL further shows that this representation is not specific to APR's
retriever: when host retrieval is held fixed, mapping the same retrieved chunks
into the shared symbolic representation reduces generator input by 60--96\% and
can substantially improve accuracy for document-style hosts. The gains remain
stable across judges, query orderings, and extraction models, while persistent
state grows approximately linearly over long query streams. APR also retains its
performance when LLM-based extraction is replaced by a local relation extractor,
showing that the approach is not tied to a particular parser.

Together, these results suggest a different design point for long-running RAG:
instead of repeatedly retrieving, serializing, and re-deriving the same information
as text, systems can maintain a compact reasoning representation that is shared
across retrieval, generation, and future queries. APR realizes this design
end-to-end, while APL provides a path for bringing the same representation and
persistent memory to existing retrieval systems.
\balance

\bibliographystyle{ACM-Reference-Format}
\bibliography{references}

\end{document}